\documentclass[11pt,a4paper]{article}
\usepackage{amsmath,amssymb,amsthm,amscd,verbatim,enumerate}
\usepackage{graphicx,subfigure}
\usepackage[utf8]{inputenc}
\usepackage[colorlinks=true,citecolor=black,urlcolor=blue]{hyperref}
\usepackage[lmargin=31mm,rmargin=31mm,bmargin=31mm,tmargin=31mm]{geometry}
\usepackage{multirow}
\usepackage{xspace}
\usepackage{cancel}

\makeatletter
\let\@fnsymbol\@alph
\makeatother

\newcommand{\pat}{\textit{path\/}}

\allowdisplaybreaks

\renewcommand{\leq}{\leqslant}
\renewcommand{\geq}{\geqslant}

\theoremstyle{plain}
\newtheorem{theorem}{Theorem}
\newtheorem{lemma}[theorem]{Lemma}

\begin{document}

\title{Analysis of Farthest Point Sampling for  Approximating Geodesics in a Graph}
\author{Pegah Kamousi\thanks{D\'{e}partement d'Informatique, Universit\'{e} Libre de Bruxelles,
Belgium, E-mail: pegahk@gmail.com} \and 
Sylvain Lazard\thanks{Inria, Loria,  Nancy, France, E-mail: sylvain.lazard@inria.fr} \and 
Anil Maheshwari\thanks{School of Computer Science, Carleton University, Canada, E-mail:
anil@scs.carleton.ca} \and 
Stefanie Wuhrer\thanks{Cluster of Excellence MMCI, Saarland University, Germany, E-mail: swuhrer@mmci.uni-saarland.de}}

\maketitle

\begin{abstract}
A standard way to approximate the distance between any two vertices  $p$ and $q$ on a mesh  is to compute, in the associated graph, a shortest path  from $p$ to  $q$ that goes through one of $k$ sources, which are well-chosen vertices. Precomputing the distance between each of the $k$ sources to all vertices of
the graph yields an efficient computation of  approximate distances between any two vertices. One standard method for choosing $k$ sources, which has been used extensively and successfully for isometry-invariant surface processing,  is the so-called  {\it Farthest Point Sampling} (FPS), which starts with a random vertex as the first source, and iteratively selects the farthest vertex from the already selected sources.

In this paper, we analyze the stretch factor $\mathcal{F}_{FPS}$ of approximate geodesics computed using FPS, which is the maximum, over all pairs of distinct vertices, of their approximated distance over their geodesic distance in the graph. We show that $\mathcal{F}_{FPS}$ can be bounded in terms of the minimal value $\mathcal{F}^*$ of the stretch factor obtained using an optimal placement of $k$ sources as $\mathcal{F}_{FPS}\leq 2 r_e^2 \mathcal{F}^*+ 2 r_e^2  + 8 r_e + 1$, where $r_e$ is the ratio of the lengths of the longest and the shortest edges of the graph. This provides some evidence explaining why farthest point sampling has been used successfully for isometry-invariant shape processing. Furthermore, we show that it is NP-complete to find $k$ sources that minimize the stretch factor.
\end{abstract}



\section{Introduction}
\label{sec:intro}

In this work, we analyze the stretch factor of approximate geodesics computed on triangle meshes or more generally in graphs. In our context, a \emph{triangle mesh} represents a discretization of a two-dimensional manifold, possibly with boundary, embedded in $\mathbb{R}^d$ for a constant dimension $d$ using a set of vertices $V$, edges $E$, and triangular faces $F$. Such a triangle mesh can be viewed as a (hyper)graph $G=(V,E,F)$ such that each edge is adjacent to one or two triangles and the triangles incident to an arbitrary vertex can be ordered cyclically around that vertex. Due to discretization artifacts, the triangle mesh may contain holes, and different parts of the surface may even intersect in the embedding. However, we assume that the graph structure $G$ is planar and connected (see Figure~\ref{fig:mesh} for an illustration). For $d=3$, this definition of a triangle mesh is commonly used in
geometry processing when analyzing models obtained by scanning real-world objects.

Given a connected planar triangle graph $G$ with $n$ vertices, where every edge has a positive length (or weight), we consider the problem of approximating the {\it geodesic} distances between pairs of vertices in $G$, where distances are measured in the graph theoretic sense. Specially, given an integer $k$, our goal is to select a set $S=\{s_1,\ldots,s_k\}$ of $k$ vertices of $G$ that minimize the {\it stretch factor}, defined as the value 
$$\max_{(p,q)\in V,\ p\neq q}~ \min_{s_i\in S}\; \frac{d(p,s_i)+d(s_i,q)}{d(p,q)},$$ 
where the function $d(.,.)$ measures the shortest geodesic distance between two vertices. Throughout this paper, we use for simplicity the notation $\max_{p,q}$ for $\max_{(p,q)\in V,\ p\neq q}$.

The problem of approximating geodesic distances on surfaces represented by triangle meshes arises when studying shapes (usually shapes embedded in $\mathbb{R}^2$ or in $\mathbb{R}^3$) that are isometric, which means that they can be mapped to each other in a way that preserves geodesics. Isometric shapes have been studied extensively recently~\cite{Bronstein2006,Elad2003,Lipman2009,Sun2009,Tevs2011} because many shapes, such as human bodies, animals, and cloth, deform in a near-isometric way as a large stretching of the surfaces would cause injury or tearing.

\begin{figure}[t]
\begin{center}
	\includegraphics[width=3.5in]{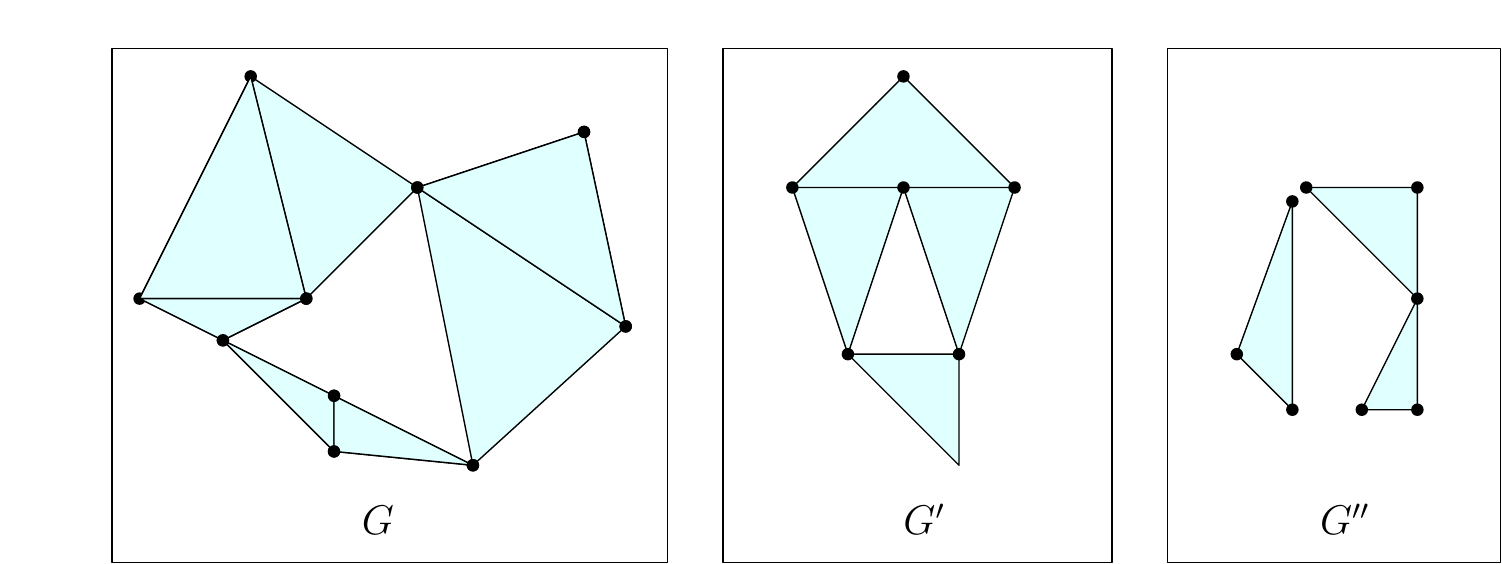}
\end{center}
\caption{The graph $G=(V,E,F)$ is a triangle mesh, while $G'$ has a non-triangular face, and $G''$ is not connected.}
\label{fig:mesh}
\end{figure}

In order to analyze near-isometric shapes, it is commonly required  to compute, on each shape, geodesics between many pairs of vertices, and to compare the  corresponding geodesics in order to compute the amount of stretching of the surface. Consider the problem of computing geodesic distances on a surface represented by a triangle mesh, where distances are measured on the graph induced by the vertices and edges of the triangulation. To solve the \emph{single-source shortest-path} (SSSP) problem on $G$, Dijkstra's algorithm~\cite{Dijkstra1959} takes $O(n \log n)$ time.\footnote{Alternatively, we can use the linear time algorithm of \cite{HKRS97} for the SSSP problem, as the underlying graph is planar.} To solve the \emph{all-pairs shortest-path} (APSP) problem, we can run Dijkstra's algorithm starting from each source point, yielding an $O(n^2 \log n)$ time algorithm. While there are more efficient methods, the APSP problem has a trivial  $\Omega(n^2)$ lower bound. In practical applications, a typical triangle mesh may contain from $100,\!000$ to $500,\!000$ vertices and, for such large meshes, it is impractical to use algorithms that take $\Omega(n^2)$ time.

To allow for a reduced complexity, instead of considering the APSP problem, we consider the problem of pre-computing a data structure that allows to efficiently \emph{approximate} the distance between \emph{any} two points. We call this the \emph{any-pair approximate shortest path} problem in the following. 

A commonly used method to solve the any-pair approximate shortest path problem is to select a set of $k$ {\it sources}, solve the SSSP problem from each of these $k$ sources, and use this information to approximate pairwise geodesic distances. Given a pair $p$ and $q$ of vertices, their shortest distance is approximated as the minimum, over all the $k$ sources, of the sum of the distances of $p$ and $q$ to a source $s_i$. This method is used to approximate the intrinsic geometry of shapes~\cite{Bronstein2006,Elad2003,Memoli2004}. A natural problem is thus to compute an optimal placement of $k$ sources that minimizes the stretch factor. We refer to this problem as the  {\it $k$-center path-dilation problem} and show that this problem is NP-complete (see Theorem~\ref{thm:hardness}).

A commonly used heuristic for selecting a set  of $k$ sources is to use {\it Farthest Point Sampling} (FPS)~\cite{gonzales_85, moenning_dodgson_03_mesh_sampling}, which starts from a random vertex and iteratively adds to $S$ a vertex that has the largest geodesic distance to its closest already picked source, until $k$ sources are picked. Given $k$ sources on a graph $G$, the distance between any two vertices $p$ and $q$ is approximated as the minimum over all $k$ sources, of the distance from $p$ to $q$ through one of the  sources.

FPS has been shown to perform well compared to other heuristics for isometry-invariant shape processing in practice~\cite{Ruggeri_etal_2010}~\cite[Chapter 3]{wuhrer2009}, which suggests that the stretch factor obtained by a FPS is small. However, to the best of our knowledge, no theoretical results are known on the quality of the stretch factor, $\mathcal{F}_{FPS}$, obtained by a FPS of $k$ sources, compared to the minimal stretch factor, $\mathcal{F}^*$, obtained by an optimal choice of $k$ sources. In this paper, we prove that \[\mathcal{F}_{FPS}\leq 2 r_e^2 (\mathcal{F}^*+ 1) + 8 r_e + 1,\] where $r_e$ is the ratio of the lengths of the longest and the shortest edges of $G$ (see Theorem~\ref{thm:ratio}). Note that this bound  holds for any arbitrary graph.

It should further be observed that if the ratio $r_e$ is large, $\mathcal{F}_{FPS}$ can be much larger than the optimal stretch factor $\mathcal{F}^*$ but, on the other hand, $\mathcal{F}^*$ is likely to be large as well. Indeed, if at least $k+1$ edges are arbitrarily small and are not ``too close'' to each other, $\mathcal{F}^*$ can be made arbitrarily large; this can be seen by considering the pairs of vertices defined by those small edges.

After discussing related work in Section~\ref{sec:related_work}, we prove our two main results,
Theorems~\ref{thm:ratio} and \ref{thm:hardness} in Sections~\ref{sec:ratio} and \ref{sec:hardness},
respectively.


\section{Related Work}
\label{sec:related_work}

Computing geodesics on polyhedral surfaces is a well-studied problem for which we refer to the recent survey by Bose et al.~\cite{Bose2011}. In this paper, we restrict geodesics to be shortest paths along edges of the underlying graph.

The FPS algorithm has been used for a variety of isometry-invariant surface processing tasks. The algorithm was first introduced for graph clustering~\cite{gonzales_85}, and later independently developed for 2D images~\cite{eldar_etal_97_image_sampling} and extended to 3D meshes~\cite{moenning_dodgson_03_mesh_sampling}. Ben Azouz et al.~\cite{BenAzouz_etal_07} and Giard and Macq~\cite{giard_macq_09_unfold} used this sampling strategy to efficiently compute approximate geodesic distances, Elad and Kimmel~\cite{Elad2003} and M\'{e}moli and Sapiro~\cite{Memoli2004} used FPS in the context of shape recognition. Bronstein et al.~\cite{Bronstein2006} and Wuhrer et al.~\cite{Wuhrer2007} used FPS to efficiently compute point-to-point correspondences between surfaces. While it has been shown experimentally that FPS is a good heuristic for isometry-invariant surface processing tasks~\cite{BenAzouz_etal_07,giard_macq_09_unfold,Elad2003,Memoli2004,Bronstein2006,Wuhrer2007}, to the best of our knowledge, the worst-case stretch of the geodesics has not been analyzed theoretically.

The problem we study is closely related to the $k$-center problem, which aims at finding $k$ centers (or sources) $s'_i$, such that the maximum distance of any point to its closest center is minimized. With the notation defined above, the $k$-center problem aims at finding $s'_i$, such that $\max_{p} \left(\min_{i} d(p,s'_i)\right)$ is minimized. This problem is $NP$-hard and FPS gives a $2$-approximation, which means that the $k$ centers $s_i$ found using FPS have the property that $\max_{p} \left(\min_{i} d(p,s_i)\right) \leq 2 \max_{p} \left(\min_{i} d(p,s'_i)\right)$~\cite{gonzales_85}. 

In the context of isometry-invariant shape processing, we are interested in bounding the stretch induced by the approximation rather than ensuring that every point has a close-by source. A related problem that has been studied in the context of networks by K\"{o}nemann et al.~\cite{koenemann_etal_04_k-center} is the {\it edge-dilation $k$-center problem}, where every point, $p$, is assigned a source, $s_p$, and the distance between two points $p$ and $q$ is approximated by the length of the path through $p, s_p, s_q,$ and $q$. The aim is then to find a set of sources that minimizes the worst stretch, and K\"{o}nemann et al. show that this problem is $NP$-hard and propose an approximation algorithm to solve the problem.

K\"{o}nemann et al.~\cite{koenemann_etal_04_k-center} also study a modified version of the above problem, which is similar to our problem. In particular, they present an algorithm for computing $k$ sources and claim that it ensures, for our problem, a stretch factor of $\mathcal{F}_K\leq 2 \mathcal{F}^*+ 1$~\cite[Theorem 3]{koenemann_etal_04_k-center}\footnote{Theorem 3 in \cite{koenemann_etal_04_k-center} is stated in a slightly different context but with the notation of that paper, considering $\pi_v=\Pi$ for every vertex $v$, the triangle inequality yields the claimed bound.}; as before,  $\mathcal{F}^*$ denotes the minimal stretch factor for the $k$-center path-dilation problem. However, we believe that their proof has gaps.
\footnote{For a given stretch factor $\alpha\leq \mathcal{F}^*$, their algorithm iteratively includes an endpoint of the shortest edge that cannot yet be approximated with a stretch of at most $2\alpha+1$ until no such edges are left. If the solution contains at most $k$ sources, a solution with stretch $2\alpha+1\leq 2 \mathcal{F}^*+ 1$ has been found. Their algorithm then essentially does a binary search on the optimal stretch factor $\mathcal{F}^*$. However, this search is done in a continuous interval without stopping criteria. Moreover,  since it is a priori possible that for any given $\alpha<\mathcal{F}^*$, their algorithm returns strictly more than $k$ sources, and that $\mathcal{F}^*$ may not be exactly reachable by dichotomy, we believe
that the stretch factor of  $\mathcal{F}_K\leq 2 \mathcal{F}^*+ 1$ is not ensured.}
It should nonetheless be stressed that our result is independent of whether this bound on
$\mathcal{F}_K$ holds. Indeed, the relevance of our
bound of $2 r_e^2 (\mathcal{F}^*+ 1) + 8 r_e + 1$ on  $\mathcal{F}_{FPS}$ is to  give some
theoretical insight on why FPS has been used successfully in heuristics
for isometry-invariant shape processing. 


\section{Approximating Geodesics with Farthest Point Sampling}
\label{sec:ratio}  

We start this section with some definitions and notation. We consider a connected graph $G$ in which the  edges have lengths from a positive and finite interval $[\ell_{min},\ell_{max}]$, and $r_e$ denotes the ratio $\ell_{max}/\ell_{min}$. We require the graph to be connected so that the distance between any two vertices is finite. In this section, we do not require the graph to satisfy any other criteria, but observe that if it is not  planar, the running time of FPS will be $O(k(m+n\log n))$,  where $m$ is the number of edges.

Given $k$ vertices (sources) $s_1,\ldots,s_k$ in the graph, let $s_p$ denote the (or a) closest source to a vertex $p$ and let $d(p,s_i,q)$ denote the shortest path length from $p$ to $q$ through any source $s_1,\ldots,s_k$, that is $\min_i (d(p,s_i)+d(s_i,q))$. Let $s^\ast_1,\cdots,s^\ast_k$ be a choice of sources that minimizes the stretch factor $\mathcal{F}^*=\max_{p,q}d(p,s_i^*,q)/d(p,q)$. Furthermore, let $s'_1,\cdots,s'_k$ be a choice of sources that minimizes $\max_p d(p,s_p)$. In other words, the set of $s^*_i$ is an optimal solution to $k$-center path-dilation problem and the set of $s'_i$ is an optimal solution to the $k$-center problem.


In this section, we prove the following theorem.
\begin{theorem}
\label{thm:ratio}
Let $s_1,\cdots,s_k$ be a set of sources returned by the FPS algorithm on a connected graph $G$ with edge lengths of ratio at most $r_e$. Then
\[\max_{p,q}\frac{d(p,s_i,q)}{d(p,q)}\leq 2r^2_e \max_{p,q} \frac{d(p,s^\ast_i,q)}{d(p,q)} + 2r^2_e+8r_e+1.\]
\end{theorem}


In order to prove this theorem, we first show a somewhat surprising property that, for any set of sources,  the stretch factor $\max_{p,q}\frac{d(p,s_i,q)}{d(p,q)}$ is realized when $p$ and $q$ are adjacent in the graph (Lemma~\ref{lem:worst-ratio}). We  use this property to bound this stretch factor in terms of $\max_p d(p,s_p)$ (Lemma~\ref{lem:tbound}).  On the other hand, we bound the stretch factor of any set of sites in terms of the stretch factor of an optimal set of  sources for the $k$-center problem (Lemma~\ref{lem:ubound}). We then combine these results to prove Theorem~\ref{thm:ratio}.

\begin{lemma}
\label{lem:worst-ratio}
For any sources $s_1,\ldots,s_k$ and any given vertex $q$ in $G$, the maximum ratio $\max_{p} \frac{d(p,s_i,q)}{d(p,q)}$ is realized for some $p$ that is adjacent to $q$ in $G$. It follows that the maximum ratio $\max_{p,q} \frac{d(p,s_i,q)}{d(p,q)}$ is realized for some $p$ and $q$ that are adjacent in $G$.
\end{lemma}

\begin{proof}
For the sake of contradiction, let $q$ be any fixed vertex and let $p$ be a non-adjacent vertex that realizes the maximum  $\max_{p} \frac{d(p,s_i,q)}{d(p,q)}$ and such that among all the vertices $p'$ that realize this maximum, the shortest path from $p$ to $q$ has the smallest number of edges.
 
Let $\tilde{p}$ be the immediate neighbor of $p$ along the shortest path from $p$ to $q$.  As before, $d(\tilde{p},s_j,q)$ denotes the shortest path length from $\tilde{p}$ to $q$ through any source $s_1,\ldots,s_k$ (we use here the notation $d(\tilde{p},s_j,q)$ instead of $d(\tilde{p},s_i,q)$ in order to avoid confusion with $d(p,s_i,q)$).  Let $\ell$ be the length of the edge $p\tilde{p}$ (see Figure~\ref{fig:path}).  We have $d(\tilde{p},s_j,q) \geq d(p,s_i,q) - \ell$. Dividing by $d(\tilde{p},q) = d(p,q) - \ell$ we get $$\frac{d(\tilde{p},s_j,q)}{d(\tilde{p},q)} \geq \frac{d(p,s_i,q)}{d(p,q)-\ell} - \frac{\ell}{d(\tilde{p},q)}.$$

On the other hand, by multiplying $d(p,q)=d(\tilde{p},q) +\ell$ by $\frac{d(p,s_i,q)}{d(p,q)d(\tilde{p},q)}$ we have $$\frac{d(p,s_i,q)}{d(p,q)-\ell} ~=~ \frac{d(p,s_i,q)}{d(p,q)} ~+~ \frac{\ell d(p,s_i,q)}{d(\tilde{p},q)\cdot d(p,q)},$$ and therefore $$\frac{d(\tilde{p},s_j,q)}{d(\tilde{p},q)} \geq \frac{d(p,s_i,q)}{d(p,q)} ~+~ \frac{\ell}{d(\tilde{p},q)} \cdot \left(\frac{d(p,s_i,q)}{d(p,q)}-1\right)\geq \frac{d(p,s_i,q)}{d(p,q)},$$ which contradicts our assumption. Indeed, either the inequality is strict and $\frac{d(p,s_i,q)}{d(p,q)}$ was not maximum, or  the equality holds and the shortest path from
$\tilde{p}$ to $q$ has fewer edges than the shortest path from $p$ to $q$.
\end{proof}

\begin{figure}[htb]
\begin{center}
	\includegraphics[width=2in]{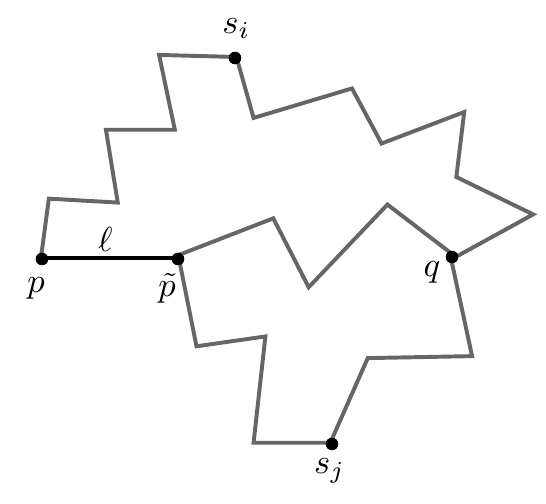}
\end{center}
\caption{For the proof of Lemma~\ref{lem:worst-ratio}.}
\label{fig:path}
\end{figure}


The property of the previous lemma that $\max_{p,q}\frac{d(p,s_i,q)}{d(p,q)}$ is realized when $p$ and $q$ are neighbors allows us to bound it as follows. 
\begin{lemma}
\label{lem:tbound}
For any sources $s_1,\ldots,s_k$, we have
$$ \frac{2}{\ell_{max}}  \max_{p}d(p,s_p)-1 \leq \max_{p,q}\frac{d(p,s_i,q)}{d(p,q)} \leq \frac{2}{\ell_{min}} \max_{p}d(p,s_p)+1.$$
\end{lemma}

\begin{proof}
For the upper bound, we have $d(p,s_i,q) \leq d(p,s_p)+d(s_p,q)\leq 2d(p,s_p)+d(p,q)$. Therefore, $\frac{d(p,s_i,q)}{d(p,q)} \leq \frac{2}{d(p,q)} d(p,s_p)+1.$ This holds for any vertices $p$ and $q$ and thus for those that realize the maximum of $\frac{d(p,s_i,q)}{d(p,q)}$. Furthermore, $d(p,q)\geq\ell_{min}$  and $d(p,s_p)\leq \max_p d(p,s_p)$. Hence, \[ \max_{p,q}\frac{d(p,s_i,q)}{d(p,q)} \leq \frac{2}{\ell_{min}} \max_{p}d(p,s_p)+1.\]

For the  the lower bound, we have by the triangle inequality that, for any $i$, $d(q,s_i) \geq d(p,s_i) - d(p,q)$. Adding $d(p,s_i)$ on both sides, we get $d(p,s_i)+d(q,s_i) \geq 2d(p,s_i) - d(p,q)$. By the definition of $s_p$, $d(p,s_i) \geq d(p,s_p)$ for any $i$, thus $d(p,s_i)+d(q,s_i) \geq 2d(p,s_p) - d(p,q)$. This holds for any $i$ and thus for the $i$ such that $d(p,s_i)+d(q,s_i)$ is minimum, hence $d(p,s_i,q) \geq 2d(p,s_p) - d(p,q)$. Dividing by $d(p,q)$, we get $\frac{d(p,s_i,q)}{d(p,q)}\geq \frac{2}{d(p,q)} d(p,s_p) - 1.$ This holds for any $p$ and $q$ and thus for the vertex $p$ that realizes the maximum of $d(p,s_p)$; let $\bar{p}$ denote such vertex. We then have that $\frac{d(\bar{p},s_i,q)}{d(\bar{p},q)}\geq \frac{2}{d(\bar{p},q)}  \max_{p}d(p,s_p)-1$.
This holds for any $q$ and in particular for the one that realizes $\max_q\frac{d(\bar{p},s_i,q)}{d(\bar{p},q)}$. By Lemma~\ref{lem:worst-ratio}, the maximum is realized for a $q$ that is adjacent to $\bar{p}$ in $G$, thus, for such a $q$, $\frac{2}{d(\bar{p},q)}\geq \frac{2}{\ell_{max}}$. It follows that 
\[\max_{p,q}\frac{d(p,s_i,q)}{d(p,q)} \geq \max_q\frac{d(\bar{p},s_i,q)}{d(\bar{p},q)}\geq \frac{2}{\ell_{max}}  \max_{p}d(p,s_p)-1.\]
\end{proof}


The following lemma bounds the path length between two vertices $u$ and $v$ passing through $s_u$ in terms of the shortest path between $u$ and $v$ through any source.

\begin{lemma}\label{lem:ineq1}
For any sources $s_1,\ldots,s_k$, and vertices $u, v$ we have $$d(u,s_u) + d(s_u,v)\leq d(u,s_i,v)+2d(u,v).$$
\end{lemma}

\begin{proof}
Denote by $s_i$ the source that realizes the minimum $d(u,s_i,v)=\min_i (d(u,s_i)+d(s_i,v))$. Since by definition $d(u,s_u)\leq d(u,s_i)$, we only have to show that $d(v,s_u)\leq d(s_i,v) +2d(u,v)$. Using the triangle inequality twice, we have \[d(v,s_u)\leq d(v,u)+d(u,s_u)\leq d(v,u)+d(u,s_i)\leq d(v,u)+d(u,v)+d(v,s_i),\] which concludes the proof.
\end{proof}


These results allow us to bound the stretch factor corresponding to the sources returned by the FPS algorithm with respect to the stretch factor corresponding to an optimal choice of sources for the $k$-center problem. 

\begin{lemma}
\label{lem:ubound}
Let $s_1,\cdots,s_k$ be a set of sources returned by the FPS algorithm and $s'_1,\cdots,s'_k$ be an optimal set of sources for the $k$-center problem. Then \[\max_{p,q} \frac{d(p,s_i,q)}{d(p,q)}\leq 2r_e\max_{u,v} \frac{d(u,s'_i,v)}{d(u,v)}+6r_e+1.\]
\end{lemma}

\begin{proof}
Since $s_1,\ldots,s_k$ is  a set of sources returned by the FPS algorithm, this choice  of sources provides a 2-approximation for the $k$-center problem compared to an optimal solution $s'_1,\ldots,s'_k$; in other words, $\max_p d(p,s_p)\leq 2 \max_p d(p,s'_p)$~\cite{gonzales_85}.

By definition, $d(p,s_i,q)$ is the minimum over all (fixed) sources $s_i$ of $d(p,s_i)+d(s_i,q)$. Thus, $d(p,s_i,q)\leq d(p,s_p)+d(s_p,q)$. Moreover, by the triangle inequality, $d(s_p,q)\leq  d(s_p,p)+d(p,q)$ thus $d(p,s_i,q)\leq 2d(p,s_p)+d(p,q)$.  One the other hand, $d(p,s_p)\leq \max_u d(u,s_u)$, which is less than or equal to $2\max_u d(u,s'_u)$ by the 2-approximation property. For clarity, denote by $u$ the vertex that realizes the maximum $\max_u d(u,s'_u)$. We then have $d(p,s_i,q)\leq 4d(u,s'_u)+d(p,q)$. 

Now, by the triangle inequality, $d(u,s'_u)\leq d(u,v)+d(v,s'_u)$ for any vertex $v$. Thus $2d(u,s'_u)\leq d(u,v)+d(v,s'_u) +d(u,s'_u)$ which implies, by Lemma~\ref{lem:ineq1}, that $2d(u,s'_u)\leq 3d(u,v)+d(u,s'_i,v)$. Thus, $d(p,s_i,q)\leq 2d(u,s'_i,v)+6d(u,v)+d(p,q)$ and \[\frac{d(p,s_i,q)}{d(p,q)}\leq 2\frac{d(u,v)}{d(p,q)}\frac{d(u,s'_i,v)}{d(u,v)}+6\frac{d(u,v)}{d(p,q)}+1.\]

This inequality holds for any distinct $p$ and $q$,  and any $v$ distinct from $u$ (recall that $u$ is fixed). Thus it holds for the vertices $p$ and $q$ that realize  $\max_{p,q}\frac{d(p,s_i,q)}{d(p,q)}$ and for the $v$ that realizes $\max_{v} \frac{d(u,s'_i,v)}{d(u,v)}$. Such a $v$ is a neighbor of $u$
by Lemma~\ref{lem:worst-ratio}, thus it satisfies $d(u,v)\leq \ell_{max}$. Since $d(p,q)\geq \ell_{min}$ for any distinct $p$ and $q$, and $\max_{v} \frac{d(u,s'_i,v)}{d(u,v)}\leq \max_{u,v} \frac{d(u,s'_i,v)}{d(u,v)}$, we get \[\max_{p,q} \frac{d(p,s_i,q)}{d(p,q)}\leq 2\frac{\ell_{max}}{\ell_{min}}\max_{u,v} \frac{d(u,s'_i,v)}{d(u,v)}+6\frac{\ell_{max}}{\ell_{min}}+1.\]
\end{proof}


This finally allows us to prove the main theorem.

\begin{proof}[Proof of Theorem~\ref{thm:ratio}]
By Lemma~\ref{lem:ubound} and using the same notation, we have \[\max_{p,q} \frac{d(p,s_i,q)}{d(p,q)}\leq 2r_e\max_{u,v} \frac{d(u,s'_i,v)}{d(u,v)}+6r_e+1.\]
Using the upper bound in Lemma~\ref{lem:tbound} on $\max_{u,v}\frac{d(u,s'_i,v)}{d(u,v)}$, we have \
$$\max_{p,q}\frac{d(p,s_i,q)}{d(p,q)} ~\leq~ 2r_e\left(\frac{2}{\ell_{min}}\max_p d(p,s'_p)+1\right)+6r_e+1.$$

By definition, $s'_1,\ldots,s'_k$ is an optimal set of sources for the $k$-center problem, that is $\arg\!\min_{s_1,\ldots,s_k} \max_p d(p,s_p)$ and thus $
\min_{s_1,\ldots,s_k} \max_p d(p,s_p)=\max_p d(p,s'_p) \leq \max_p d(p,s^*_p) $.

We now apply the lower bound of Lemma~\ref{lem:tbound} to $s^*_1,\ldots,s^*_k$ which gives \[\frac{2}{\ell_{max}}  \max_{p}d(p,s^*_p)-1 \leq \max_{p,q}\frac{d(p,s^*_i,q)}{d(p,q)}\] and  thus 
\[\begin{split}
\max_{p,q}\frac{d(p,s_i,q)}{d(p,q)}&\leq 2r_e\left(\frac{\ell_{max}}{\ell_{min}}(\max_{p,q} \frac{d(p,s^\ast_i,q)}{d(p,q)}+1)+1\right)+6r_e+1\\
&\leq 2r^2_e \max_{p,q} \frac{d(p,s^\ast_i,q)}{d(p,q)} + 2r^2_e+8r_e+1.
\end{split}\]
\end{proof}


\section{The Complexity of $k$-Center Path-Dilation Problem}
\label{sec:hardness}

In this section we consider the complexity of the $k$-center path-dilation problem on triangle graphs, i.e., computing an optimal set of sources that minimizes the stretch factor. The following theorem shows that the decision version of this problem is NP-complete for triangle graphs. Note that this directly yields the NP-completeness for arbitrary graphs (since proving that the problem is in NP is trivial).

\begin{theorem}
\label{thm:hardness}
Given a triangle graph $G=(V,E,F)$, an integer $k$, and a real value $\xi$, it is NP-complete to determine whether there exists  a set $S=\{s_1,\ldots,s_k\}$ of $k$ sources such that the stretch factor $\max_{p,q}\min_i\frac{d(p,s_i)+d(s_i,q)}{d(p,q)}$ is at most $\xi$.
\end{theorem}

Note that the  problem is in NP since, for any set of $k$ sources, the stretch factor can be computed in polynomial time. To show the hardness, we provide a reduction from the decision problem related to finding a minimum cardinality vertex cover on planar graphs of maximum vertex degree three~\cite{GJ77}. The first step of the reduction uses the following well-known result on embedding planar graphs in integer grids~\cite{Valiant81}.

\begin{lemma}
\label{lem:embed}
A planar graph $G = (V, E)$ with maximum degree $4$ can be embedded in the plane using $O(|V|^2)$ area in such a way that its nodes have integer coordinates and its edges are drawn as polygonal line segments that lie on the integer grid (i.e, every edge  consists of one or more line segments that lie on lines of the form  $x=i$ or $y=j$, where $i$ and $j$ are integers).
\end{lemma}

\begin{figure}[t]
\begin{center}
	\includegraphics[width=2in]{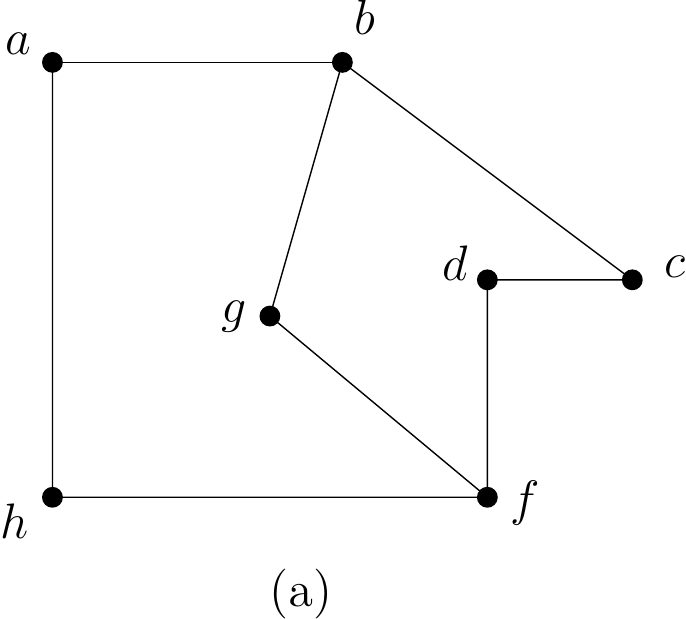}	\hspace{0.5cm}
	\includegraphics[width=2.2in,page=2]{embed}\\
\vspace{0.6cm}
	\includegraphics[width=2in]{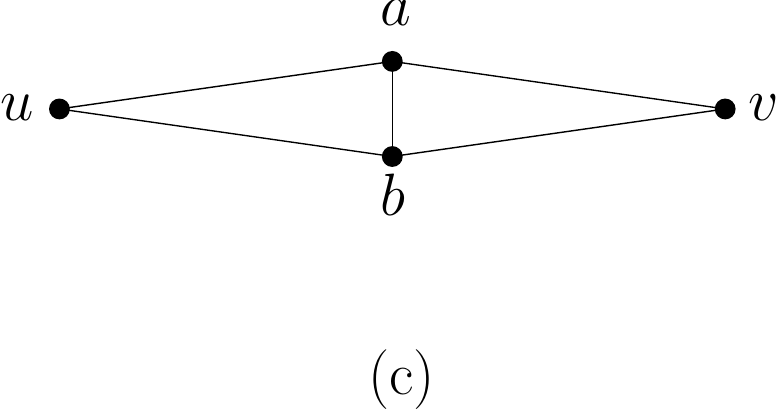}	\hspace{0.5cm}
	\includegraphics[width=2.2in,page=2]{replace}
\end{center}
\caption{(a) A planar graph $G$ and (b) the grid-embedding of $G_r$. The gadget $\varrho$ replacing the edges of $G_r$ and (d) the resulting graph $G'$.}
\label{fig:embed}
\end{figure}

Consider a planar graph $G$ with maximum degree 3, and let $G_r$ be a planar embedding of $G$ according to Lemma~\ref{lem:embed}, to which we have added, on each edge $e\in E$, an {\it even} number of $2k_e$  {\it auxiliary} nodes  with half-integer coordinates and such that every resulting edge in $G_r$ has length  $1$ or $1/2$. (We consider the half-integer grid so that we can ensure that we add an {\it even} number of auxiliary nodes on every edge of $G$ so that every resulting edge in $G_r$ has length at most 1.) Please refer to Figure~\ref{fig:embed}(a,b) for an illustration. For an edge $uv\in E$, we let
$\pat(u,v)$ denote the path in $G_r$ replacing the edge $uv$. The endpoints of the paths (i.e., the nodes that are not auxiliary), are called {\it regular} nodes. Finally, let $m=\sum_{e\in E} k_e$. We have the following lemma.

\begin{lemma}
\label{lem:vc}
$G$ has a vertex cover of size $k$ if and only if $G_r$ has a vertex cover of size $k+m$.
\end{lemma}

\begin{proof}
Any vertex cover $C$ of $G$ with size $k$ can be extended to a vertex cover of size $k+m$ in $G_r$ by including every other auxiliary node on $\pat(u,v)$, for each edge $uv \in E$.

Now let $C_r$ be a vertex cover for $G_r$ of size $m+k$, and suppose there exists a path, $\pat(u,v)$, such that neither $u$ nor $v$ belongs to $C_r$. Then at least $k_{uv}+1$ auxiliary nodes from $\pat(u,v)$ must belong to $C_r$ in order to cover all the edges of this path. However, by using only $k_{uv}$ auxiliary nodes from $\pat(u,v)$ and adding $u$ or $v$ to $C_r$, we still have a vertex cover of the same size, which now contains one of the endpoints of $\pat(u,v)$. Continuing this way, we can construct a vertex cover $C_r$ of size $k+m$ for $G_r$, which includes at least one endpoint from each $\pat(u,v)$, for all $uv \in E$. Therefore, $C_r$ is a vertex cover for $G$, when restricted to the nodes of $G$ (regular nodes). Since a minimum of $k_{uv}$ auxiliary nodes are needed to cover any path $\pat(u,v)$, $uv\in E$ (even if both $u$ and $v$ belong to the vertex cover), the number of regular nodes selected is at most $k$. This concludes the proof.
\end{proof}
 
Finally, we replace each edge in $G_r$ with a copy of the gadget $\varrho$ illustrated in Figure~\ref{fig:embed}(c), and denote the resulting graph by $G'$ (see figure~\ref{fig:embed}(d)). (We note that each copy is scaled, while maintaining the proportions, to match with the length of the edge it replaces.)

\begin{proof}[Proof of Theorem~\ref{thm:hardness}]
Consider the graph $G'$, constructed as above from a planar graph $G$ with maximum degree $3$ and with a gadget such that $\frac{|au|+|bu|}{|ab|}=\xi\geq 3$. The graph $G'$ can be seen as  a union of triangles, and  it thus  a triangle mesh. We prove in the following that $G$ has a vertex cover of size $k$ if and only if $G'$ has $k+m$ sources such that its stretch factor is at most $\xi$. Hence, the vertex cover problem can be reduced in polynomial time to the problem at hand, which concludes the proof.

We first show that if $G$ has a vertex cover of size $k$, then there is a set of $k+m$ sources in $G'$ whose stretch factor is $\xi$. If $G$ has a vertex cover of size $k$, then $G_r$ has a vertex cover of size $k+m$ by Lemma~\ref{lem:vc}. Recall that this vertex cover of $G_r$ can be obtained from the vertex cover of $G$ by adding every other auxiliary node on each edge of $G$. Let this vertex cover of $k+m$ nodes be the choice of sources in $G'$. Consider a pair $p$ and $q$ of nodes. We consider three cases:

\begin{enumerate}
\item  $p$ and $q$ belong to the same gadget (the same copy of $\varrho$). Let $u$, $v$, $a$ and $b$ denote the nodes of this gadget, as illustrated in Figure~\ref{fig:embed}(c), and suppose, without loss of generality, that $u$ is selected as a source. Then, $\frac{d(p,u)+d(u,q)}{d(p,q)}$ is equal to 1 if $p$ or $q$ coincides with $u$,  it is by definition equal to $\xi$ if $(p,q)=(a,b)$, and it is equal to $3$ if $p=v$ and $q=a$ or $q=b$ (see Figure~\ref{fig:gadget2}(a)). Since $\xi\geq 3$ by definition of the gadget, the maximum of $\min_i\frac{d(p,s_i^*)+d(s_i^*,q)}{d(p,q)}$, over all pairs $(p,q)$ in a gadget, is $\xi$.

\begin{figure}[t]
\begin{center}
	\includegraphics[width=11cm]{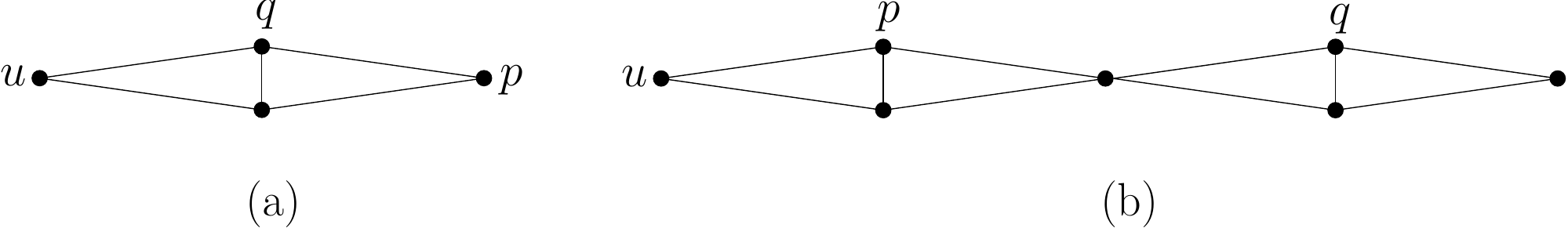}
\end{center}
\caption{For the proof of Theorem~\ref{thm:hardness}.}
\label{fig:gadget2}
\end{figure}

\item $p$ and $q$ belong to two adjacent gadgets $\varrho_1$ and $\varrho_2$. Let $\{u_1,a_1,b_1,v_1,a_2,b_2, v_2\}$ denote the nodes of the two gadgets. If the node $v_1$ which belongs to both gadgets is selected as a source then, by symmetry, the analysis of Case $1$ yields the same bound of $\xi$. Otherwise, $u_1$ and $v_2$ are sources and,  for any two nodes from two different gadgets, the maximum ratio is $2$ (the ratio is 1 if $p$ or $q$ is   one
of the sources and the ratio is 2 in the other cases; see Figure~\ref{fig:gadget2}(b)). Therefore, again the maximum ratio is at most $\xi$.

\item $p$ and $q$ belong to neither the same gadget nor to two adjacent gadgets. In this case, at least one of the nodes in a shortest path from $p$ to $q$ is selected as a source, and hence their approximate shortest path equals the geodesic shortest path and $\min_i\frac{d(p,s_i^*)+d(s_i^*,q)}{d(p,q)}=1$.
\end{enumerate}

We thus proved that the stretch factor of $G'$, for the selected $k+m$ sources, is $\xi$.

Conversely, we show  that if  $G'$ has $k+m$ sources such that its stretch factor is at most $\xi$, then  $G$ has a vertex cover of size $k$. Every  gadget must contain at least a source since, otherwise, the vertices $a$ and $b$ of a gadget with  no source are such that $\min_i\frac{d(a,s_i^*)+d(s_i^*,b)}{d(a,b)}>\xi$. For every gadget in $G'$, if vertex $u$ or $v$ is a source, we select the corresponding vertex in $G_r$, and if vertex  $a$ or $b$  is a source, we select any endpoint of the edge of  $G_r$ corresponding to the gadget. Then, at least one vertex from each edge of the graph $G_r$ is selected. Hence, $G_r$ has a vertex cover of size $k+m$, which implies that $G$ has a vertex cover of size $k$ by Lemma~\ref{lem:vc}. This completes the proof.
\end{proof}


\section{Conclusions}

We analyzed the stretch factor $\mathcal{F}_{FPS}$ of approximate geodesics computed as distances through at least one of a set of $k$ sources found using farthest point sampling. We  showed that $\mathcal{F}_{FPS}$ can be bounded by $2r^2_e \mathcal{F}^*+ 2r^2_e+8r_e+1$, where $\mathcal{F}^*$ is  stretch factor obtained using an optimal placement of the sources and $r_e$ is the ratio of the lengths of the longest and the shortest edges in the graph. Furthermore, we showed that it is NP-complete to find such an optimal placement of the sources. Note that in many practical applications $r_e\approx 1$, which  gives some evidence explaining why farthest point sampling has been used successfully for isometry-invariant surface processing.


\section*{Acknowledgments}
This research was initiated at Bellairs Workshop on Geometry and Graphs, March 10--15, 2013. The authors are grateful to Prosenjit Bose, Vida Dujmovic, Stefan Langerman, and Pat Morin for organizing the workshop and to the other workshop participants for providing a stimulating working environment.


\end{document}